\documentclass{article}
\usepackage{graphicx} 
\usepackage{amsmath,amssymb,amsthm,graphicx,comment}
\usepackage[T2A]{fontenc}
\usepackage[russian, english]{babel}
\usepackage[left=2.5cm,right=2.2cm,top=2.5cm,bottom=3.5cm]{geometry}
\usepackage{forest}
\usetikzlibrary{arrows.meta}
\usepackage[upright]{fourier}
\usepackage{tkz-graph}
\usetikzlibrary{shapes.geometric, arrows}
\usetikzlibrary{automata,positioning}

\let\le\leqslant
\let\ge\geqslant

\let\geq\geqslant

\newtheorem{theorem}{Theorem}

\newtheorem*{remark*}{Remark}

\theoremstyle{definition}
\newtheorem{definition}{Definition}
\newtheorem*{hypothesis*}{Hypothesis}
\title{Nondeterministic state complexity of square root}
\author{Sergey Onishchenko\footnote{Saint Petersburg University, Faculty of Mathematics and Computer Science, 7/9 Universitetskaya nab., St. Petersburg, 199034 Russia. E-mail: dimer03@mail.ru}}
\date{April 9, 2026}

\begin{document}

\maketitle

\begin{abstract}
We investigate the nondeterministic state complexity of the square-root operation $\sqrt{L}=\{\,w \mid ww\in L\,\}$ on regular languages represented by nondeterministic finite automata. For an
$n$-state NFA accepting $L$, it was previously known that $\sqrt{L}$ can be
accepted by an NFA with at most $n^{3}$ states, while the best lower bound was only (n-1)(n-2)(n-3). In this paper, we close this gap completely
and prove that $n^{3}$ states are sufficient and necessary in the worst case.
\end{abstract}


\section{Introduction}
The (non)deterministic state complexity of a regular language $L$ is the number of states in the
minimal (non)deterministic finite automaton recognizing $L$. The square root of a language is defined as $\sqrt{L} := \{w \in \Sigma ^*| ww \in L\}$.

We consider the following problem: NFA $A$ with n states over the alphabet $\Sigma$ recognizes the language L. How many states might be needed to recognize the language $\sqrt{L} := \{w \in \Sigma ^*| ww \in L\}$?

There has been much
interest in the study of state complexity of operations which preserve
regularity, or, in other words, in the descriptional complexity of operations
on languages represented by finite automata. Yu et al. \cite{YuZhuangSalomaa1994} established classical tight bounds for several basic DFA operations on
regular languages. Salomaa et al. \cite{SalomaaWoodYu2004} gave a detailed study of reversal. They confirmed the worst-case bound and described classes of automata and
languages attaining this maximal blow-up. Domaratzki and Okhotin \cite{DomaratzkiOkhotin2009} studied powers $L^k$ of a regular language. For fixed $k \geq 2$ and
an $n$-state DFA for $L$, they proved the general asymptotic bound. Also Cui et al. \cite{CuiGaoKariYu2012} completed the study of combined operations with two basic operations
among union, intersection, catenation, star, and reversal.

The problem of the state complexity of the square-root operation in the deterministic case was first raised by Maslov \cite{M}. He proved the upper bound $n^n$. Then Caron et al. \cite{CCLP} proved the exact bound, that turned out to be a little smaller---namely $n^n-{n \choose 2}$.

When the operands are given by
nondeterministic finite automata (NFAs), the corresponding questions often
become substantially more delicate, which was represented, for example, by Holzer and Kutrib \cite{HK}. Nondeterminism may yield exponential
savings in representation size, but the effect of language operations on this
succinctness has not been fully investigated yet. In this sense, the
operational complexity of NFAs remains one of the central themes of
descriptional complexity. Jir{\'a}skov{\'a} and Okhotin \cite{JiraskovaOkhotin2008} studied the cyclic shift operation. For NFAs, they established the tight bound for the nondeterministic state
complexity of cyclic shift. Go{\v{c}} et al. \cite{GocPalioudakisSalomaa2014} studied nondeterministic state complexity of proportional removals. They considered polynomial removals and proved an upper bound
for languages represented by $n$-state NFAs. Matching lower bounds were obtained
for important classes of removal functions.
Han et al. \cite{HanSalomaaWood2009} discusses the operations for prefix-free regular languages. The operations considered
include catenation, union, intersection, Kleene star, reversal, and
complementation. The results show how prefix-freeness changes the usual NFA
state-complexity bounds for general regular languages.
Hospod{\'a}r et al. \cite{HospodarJirasekJiraskovaSebej2025} studied NFA-to-DFA trade-offs for regular operations: the operands are
given by NFAs, while the result must be represented by a DFA. They obtained tight
bounds for many operations, including complementation, union, symmetric
difference, intersection, difference, quotients, reversal, star, concatenation,
cut, shuffle, and square.

It is well known that for any language $L$, which is recognized by an NFA with $n$ states, language $\sqrt{L}$ is recognized by an NFA with $n^3$ states---this exercise is often found in university courses on automata theory.

The best lower bound was proved by Okhotin and Salomaa \cite{OS} and equal to $(n-1)(n-2)(n-3)$.
\\
In this paper, we are going to prove the exact lower bound $n^3$: for every $n$ we provide an example of NFA with $n$ states, which recognizes language $L$, such that $\sqrt{L}$ is not recognized by any NFA with fewer than $n^3$ states.

\section{Preliminaries}

Let $\Sigma$ be a finite alphabet. We write $\Sigma^*$ for the free monoid generated by $\Sigma$, and $\varepsilon$ for the empty word.

A \emph{nondeterministic finite automaton} (NFA) is a quintuple
\[
A=(Q,\Sigma,\delta,Q_0,F),
\]
where $Q$ is a finite set of states, $\Sigma$ is the input alphabet, $Q_0 \subset Q$ is the set of initial states, $F \subseteq Q$ is the set of accepting states, and
\[
\delta \colon Q \times \Sigma \to 2^Q
\]
is the transition function.

The transition function is extended in the standard way to a mapping
\[
\delta \colon Q \times \Sigma^* \to 2^Q
\]
by setting
\[
\delta(q,\varepsilon)=\{q\}
\]
for every $q \in Q$, and
\[
\delta(q,wa)=\bigcup_{p \in \delta(q,w)} \delta(p,a)
\]
for every $q \in Q$, $w \in \Sigma^*$, and $a \in \Sigma$.

The language recognized by $A$ is
\[
L(A)=\{\, w \in \Sigma^* \mid \delta(q_0,w)\cap F \neq \emptyset \,\}.
\]

We shall use the fooling set technique to derive lower bounds on the nondeterministic state complexity.

\begin{definition}
Let $L \subseteq \Sigma^*$. A set
\[
S=\{(x_i,y_i)\mid 1 \le i \le m\}\subseteq \Sigma^* \times \Sigma^*
\]
is called a \emph{fooling set} for $L$ if the following conditions are satisfied:
\begin{enumerate}
    \item $x_i y_i \in L$ for every $1 \le i \le m$;
    \item for every pair of distinct indices $i,j$, at least one of the words $x_i y_j$ and $x_j y_i$ does not belong to $L$.
\end{enumerate}
\end{definition}

The following theorem was proved by Jean-Camille Birget \cite{B}.

\begin{theorem}
Let $L \subseteq \Sigma^*$. If $L$ admits a fooling set of cardinality $m$, then every NFA recognizing $L$ has at least $m$ states.
\end{theorem}

\newpage


\section{The upper bound}
For the sake of completeness, here we will present that well-known construction. The idea of our lower-bound construction will be based on it.

For every NFA $A = (Q,\Sigma,\delta,Q_0,F)$ with $n$ states, which recognizes language $L$, we will provide a NFA $B$ with $n^3$ states, which recognizes language $\sqrt L$.

\begin{definition}\label{main}
Let us define an NFA $B=(\Sigma, Q', Q'_0, \delta', F')$'.

\[
Q' = Q \times Q \times Q.
\]

Obviously, $|Q'| = n^3$.

\[
Q'_0=\{(p,q_0,p)\mid p\in Q, q_0 \in Q_0\}.
\]

\[
F'=\{(p,p,f)\mid p\in Q,\ f\in F\}.
\]

For every \(a\in \Sigma\), define the transition relation by
\[
(p,q,r)\xrightarrow{a}_{\delta'}(p,q',r')
\]
if and only if
\[
q'\in \delta(q,a)
\qquad\text{and}\qquad
r'\in \delta(r,a).
\]
\end{definition}

\begin{theorem}
NFA $B$ recognizes $\sqrt L$.
\end{theorem}

The proof of this theorem is widely known and is not particularly difficult, so we will not present it here.

\section{The lower bound}
This is our main result. For every $n\ge 6$ we provide an example of NFA $A_n$ with $n$ states, which recognizes language $L$, such that $\sqrt{L}$ is not recognized by any NFA with fewer than $n^3$ states.

\begin{definition}\label{main}
\begin{center}
Let us define an NFA $A_n=(\Sigma, Q, Q_0, \delta, F)$ for any $n\ge 6$.
\end{center}
\[
Q=\{0,1,2,3,4,5,\ldots,n-1\}.
\]

\[
Q_0=\{0,1,2\},\qquad F=\{3,4,5\}.
\]

\[
\Sigma=\{a_X,b_X|X=(p,q,r) \in Q^3\}.
\]

\[
\delta(a_X,l)=q, \qquad \delta(a_X,p)=r, \qquad \delta(b_X,q)=p,\qquad \delta(b_X,r)=m,
\]
\[
l=
\begin{cases}
0, & p\notin\{0,1\},\\
1, & p=0,\\
2, & p=1,
\end{cases}
\qquad
m=
\begin{cases}
3, & p\notin\{3,4\},\\
4, & p=3,\\
5, & p=4.
\end{cases}
\]
\\
\begin{center}
Let $L$ be the language recognized by $A$.
\end{center}
\end{definition}

\begin{theorem}\label{main}
$\sqrt{L}$ is not recognized by any NFA with fewer than $n^3$ states.
\end{theorem}

\begin{proof}[Proof of Theorem 3.]
$S:=\{(a_X,b_X)|X=(p,q,r) \in Q^3\}$.
$|S|=n^3$, so it is enough to prove that S - fooling set for $\sqrt{L}$.

$a_{X_1}b_{X_2} \in \sqrt{L} \iff a_{X_1}b_{X_2}a_{X_1}b_{X_2} \in L$. Suppose the A accepts $a_{X_1}b_{X_2}a_{X_1}b_{X_2}$, when started in some state $i \in \{0, 1, 2\}$. We represent as a tree all runs that can be accepting. The edges are annotated with the conditions under which a run could have traversed that edge.

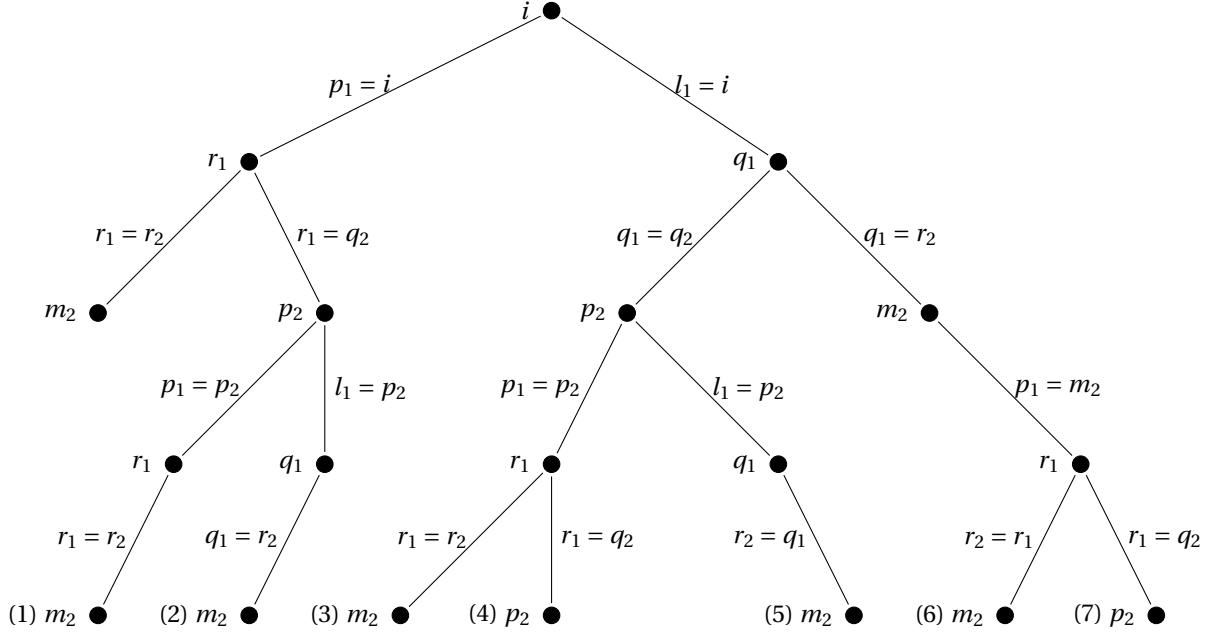
\tikzset{Bullet/.style={fill = black, draw, color=#1, outer sep = 2, circle, minimum size = 3pt, scale = 0.75}}
    \begin{tikzpicture}
    \node[Bullet=black,label=left : {$i$}] (i) at (0, 6){};
    \node[Bullet=black,label=left : {$r_1$}] (r_1) at (-4, 4){};
    \node[Bullet=black,label=left : {$q_1$}] (q_1) at (3, 4){};
    \node[Bullet=black,label=left : {$m_2$}] (m_2) at (-6, 2){};
    \node[Bullet=black,label=left : {$p_2$}] (p_2) at (-3, 2){};
    \node[Bullet=black,label=left : {$p_2$}] (p2) at (1, 2){};
    \node[Bullet=black,label=left : {$m_2$}] (m2) at (5, 2){};
    \node[Bullet=black,label=left : {$r_1$}] (r1) at (-5, 0){};
    \node[Bullet=black,label=left : {$q_1$}] (q1) at (-3, 0){};
    \node[Bullet=black,label=left : {$r_1$}] (r^1) at (0, 0){};
    \node[Bullet=black,label=left : {$q_1$}] (q^1) at (3, 0){};
    \node[Bullet=black,label=left : {$r_1$}] (r@1) at (7, 0){};
    \node[Bullet=black,label=left : {(1) $m_2$}] (m^2) at (-6, -2){};
    \node[Bullet=black,label=left : {(2) $m_2$}] (m@2) at (-4, -2){};
    \node[Bullet=black,label=left : {(3) $m_2$}] (m-2) at (-2, -2){};
    \node[Bullet=black,label=left : {(4) $p_2$}] (p^2) at (0, -2){};
    \node[Bullet=black,label=left : {(5) $m_2$}] (m*2) at (4, -2){};
    \node[Bullet=black,label=left : {(6) $m_2$}] (m+2) at (6, -2){};
    \node[Bullet=black,label=left : {(7) $p_2$}] (p@2) at (8, -2){};

    \draw (i)--node[left]{$p_1=i$}(r_1);
    \draw (r_1)--node[left]{$r_1=r_2$}(m_2);
    \draw (r_1)--node[right]{$r_1=q_2$}(p_2);
    \draw (p_2)--node[left]{$p_1=p_2$}(r1);
    \draw (p_2)--node[right]{$l_1=p_2$}(q1);
    \draw (r1)--node[left]{$r_1=r_2$}(m^2);
    \draw (q1)--node[left]{$q_1=r_2$}(m@2);
    \draw (i)--node[right]{$l_1=i$}(q_1);
    \draw (q_1)--node[left]{$q_1=q_2$}(p2);
    \draw (q_1)--node[right]{$q_1=r_2$}(m2);
    \draw (p2)--node[left]{$p_1=p_2$}(r^1);
    \draw (p2)--node[right]{$l_1=p_2$}(q^1);
    \draw (m2)--node[right]{$p_1=m_2$}(r@1);
    \draw (r^1)--node[left]{$r_1=r_2$}(m-2);
    \draw (r^1)--node[right]{$r_1=q_2$}(p^2);
    \draw (q^1)--node[left]{$r_2=q_1$}(m*2);
    \draw (r@1)--node[left]{$r_2=r_1$}(m+2);
    \draw (r@1)--node[right]{$r_1=q_2$}(p@2);
\end{tikzpicture}
\\
So, $a_{X_1}b_{X_2} \in \sqrt{L}$ in seven cases:
\\
1) $p_1=p_2=i, r_1=r_2=q_2$
\\
2) $p_1=i, p_2=l_1, r_1=q_2, q_1=r_2$
\\
3) $p_1=p_2, q_1=q_2, r_1=r_2$
\\
4) $p_1=p_2 \in F, r_1=q_1=q_2$
\\
5) $p_2=l_1, q_1=q_2=r_2$
\\
6) $p_1=m_2, q_1=r_1=r_2$
\\
7) $p_1=m_2, r_1=q_2, q_1=r_2, p_2 \in F$
\\
Therefore, for any $X \in Q^3$ the string $a_Xb_Xa_Xb_X$ is accept by $A$---case 3 represents it.
\\
Suppose that $S$ is not a fooling set, i.e., $\exists X_3 \neq X_4: a_{X_3}b_{X_4},a_{X_4}b_{X_3} \in \sqrt{L}$. Then, each of these two strings satisfies one of the cases 1-2 or 4-7.
\\
Suppose that $p_3=p_4$ - this allows us to consider only cases 1 and 4, because, if case 2 or 5 is satisfied, $p_2=l_1 \neq p_1$, and, if case 6 or 7 is satisfied, $p_1=m_2 \neq p_2$.
\\
Let $p_3=p_4 \in Q_0$, then both for $a_{X_3}b_{X_4}$ and for $a_{X_4}b_{X_3}$ case 1 occurs, because, if case 4 is satisfied, $p_1=p_2 \in F$ and $F \cap Q_0 = \emptyset$.

If case 1 satisfied, $r_1=r_2=q_2$. Case 1 satisfied for $a_{X_3}b_{X_4}$, therefore $r_3=r_4=q_4$. Case 1 satisfied for $a_{X_4}b_{X_3}$, therefore $r_4=r_3=q_3$. So, $r_3=r_4=q_3=q_4 \Rightarrow X_3=X_4$.
\\
Now let $p_3=p_4 \notin Q_0$, then both for $a_{X_3}b_{X_4}$ and for $a_{X_4}b_{X_3}$ case 4 occurs, because, if case 1 is satisfied, $p_1=p_2 \in Q_0$.

If case 4 satisfied, $r_1=q_1=q_2$. Case 4 satisfied for $a_{X_3}b_{X_4}$, therefore $r_3=q_3=q_4$. Case 4 satisfied for $a_{X_4}b_{X_3}$, therefore $r_4=q_4=q_3$. So, $r_3=r_4=q_3=q_4 \Rightarrow X_3=X_4$.
\\
Suppose that $p_3 \neq p_4$---this allows us to consider only cases 2, 5, 6 and 7, because, if case 1 or 4 is satisfied, $p_1=p_2$.
\\
Let string $a_{X_3}b_{X_4}$ satisfy one of the cases 2 and 5. In these cases $p_2=l_1 \in Q_0$, therefore $p_4 \in Q_0$. In cases 6 and 7 we have $p_1=m_2 \notin Q_0$ (because $m_2 \in F$, $F \cap Q_0 = \emptyset$), so string $a_{X_4}b_{X_3}$ also satisfy one of the cases 2 and 5 (since we know, that $p_4 \in Q_0$). So, strings $a_{X_3}b_{X_4}$ and $a_{X_4}b_{X_3}$ both correspond to the pair of cases 2 and 5. In these cases $p_2=l_1$. Therefore, $p_3=l_4$ and $p_4=l_3$. But conditions $p_3=l_4$ and $p_4=l_3$ are inconsistent by definition of $l$.

Let string $a_{X_3}b_{X_4}$ satisfy one of the cases 6 and 7. In these cases $p_1=m_2 \in F$, therefore $p_3 \in F$. In cases 2 and 5 we have $p_2=l_1 \notin F$ (because $l_1 \in Q_0$, $F \cap Q_0 = \emptyset$), so string $a_{X_4}b_{X_3}$ also satisfy one of the cases 6 and 7 (since we know, that $p_3 \in F$). So, strings $a_{X_3}b_{X_4}$ and $a_{X_4}b_{X_3}$ both correspond to the pair of cases 6 and 7. In these cases $p_1=m_2$. Therefore, $p_3=m_4$ and $p_4=m_3$. But conditions $p_3=m_4$ and $p_4=m_3$ are inconsistent by definition of $m$.
\\
Contradiction.
\end{proof}
\section{Acknowledgement}
The work of S. Onishchenko was performed at the Saint Petersburg Leonhard Euler International Mathematical Institute and supported by the Ministry of Science and Higher Education of the Russian Federation (agreement no. 075–15–2025–343).


\end{document}